\documentclass[aps,prl,twocolumn,superscriptaddress]{revtex4}

\usepackage{amsmath}
\usepackage{times}
\usepackage{hyperref}
\usepackage{cleveref}
\usepackage{autonum}
\usepackage{graphicx}
\usepackage{dcolumn}
\usepackage{enumerate,amsthm,amssymb,color}
\usepackage{bbm}
\usepackage{bm}
\usepackage{color}
\usepackage[french, english]{babel}
\usepackage{synttree}
\usepackage{multirow}
\usepackage{xfrac}
\usepackage{fixltx2e}

\newcommand{\ket}[1]{|#1\rangle}

\newcommand{\COMMENT}[1]{}

\def\defeq{\mathrel{\mathop:}=} 

\newcommand{\je}[1]{{\color{cyan} #1}}

\newtheorem{proposition}{Proposition}
\newtheorem{theorem}{Theorem}
\newtheorem{corollary}{Corollary}

\newtheorem{define}{Definition}
\newtheorem{lemma}{Lemma}

\newtheorem*{observation*}{Observation}

\newtheorem{innercustomcor}{Corollary}
\newenvironment{customcor}[1]
{\renewcommand\theinnercustomcor{#1}\innercustomcor}
{\endinnercustomcor}

\begin{document}

	\title{Quantum network routing and local complementation}
	
	\author{F.\ Hahn}\affiliation{Dahlem Center for Complex Quantum Systems, Freie Universit{\"a}t Berlin, 14195 Berlin, Germany}
	\author{A.\ Pappa}\affiliation{Department of Physics and Astronomy, University College London, Gower Street, London WC1E 6BT, UK}
	\author{J.\ Eisert}\affiliation{Dahlem Center for Complex Quantum Systems, Freie Universit{\"a}t Berlin, 14195 Berlin, Germany}
	
	\date{\today}

	\begin{abstract}
Quantum communication between distant parties is based on suitable instances of shared entanglement. For efficiency reasons, in an anticipated quantum network beyond point-to-point communication, it is preferable that many parties can communicate simultaneously over the underlying infrastructure; however, bottlenecks in the network may cause delays. Sharing of multi-partite entangled states between parties offers a solution, allowing for parallel quantum communication. Specifically for the two-pair problem, the butterfly network provides the first instance of such an advantage in a bottleneck scenario. The underlying method differs from standard repeater network approaches in that it uses a graph state instead of maximally entangled pairs to achieve long-distance simultaneous communication. We will demonstrate how graph theoretic tools, and specifically local complementation, help decrease the number of required measurements compared to usual methods applied in repeater schemes. We will examine other examples of network architectures, where deploying local complementation techniques provides an advantage. We will finally consider the problem of extracting graph states for quantum communication via local Clifford operations and Pauli measurements, and discuss that while the general problem is known to be {\tt NP}-complete, interestingly, for specific classes of structured resources, polynomial time algorithms can be identified.
	\end{abstract}
	
	\maketitle
	
	\subsection{Introduction}
	Quantum communication schemes over optical networks necessarily 
	suffer from transmission losses and errors. For this reason, in order to achieve the vision of secure quantum communication over arbitrary distances, several schemes have been proposed that
	are based on entanglement swapping and purification 
	\cite{Briegel98,RepeaterLukin,van2006hybrid,Acin07,Roadmap}. However, such existing ``quantum repeater'' approaches are based on sharing and manipulating close to
	maximally entangled ``EPR'' pairs between the nodes. A lot of emphasis
	has been put onto identifying efficient ways of achieving this task
	\cite{RepeaterLukin,zwerger2012measurement,van2006hybrid,TeleportationReview}, 
	amounting to challenging prescriptions. Yet, for multi-partite quantum networks going beyond point-to-point achitectures, much less is known about how to meaningfully make use of and manipulate resources. 
	This is particularly unfortunate since a number of protocols have been devised for tasks like secret sharing \cite{SecretSharing,bell2014experimental}, quantum voting \cite{PhysRevA.95.062306} and quantum conference key agreement \cite{PhysRevA.97.022307,ConferenceKey,EppingA}, that exploit
	the genuine multi-partite character of a quantum network, having the vision of a
	quantum internet in mind \cite{QuantumInternet}. In fact, one could argue that the true potential
	of quantum communication is expected to lie in such multi-partite applications beyond point-to-point architectures.
	
	Specifically in multi-partite quantum networks, it could well be preferable that the involved processes are run offline, i.e., before a request for communication is received. However, methods like the ones described in Ref.~\cite{SMIKW16} require big quantum memories, as well as a high channel capacity. Consequently, network efficiency is limited by the memory capacities of the quantum repeater stations \cite{zwerger2018long}, as well as by possible bottlenecks imposed by the quantum network architecture. What is more, in many applications, multi-partite resources are required in the first place. 
	In this context, new questions of \emph{quantum routing} emerge.
	We use the term quantum routing as referring to the task of manipulating entangled resources in multi-partite quantum networks between arbitrary nodes, not necessarily making use of local knowledge only, as is common in classical routing, but allowing for global classical communication. 	The key question in this framework is how to optimally establish communication between distant nodes using the intermediate nodes of a quantum network.
	
	In this work, we consider alternative ways for sharing entanglement between distant nodes of a network that have 
	favorable features both with respect to memory and channel capacity. We start from the same setting 
	where nodes that are connected with physical optical links share close to maximally entangled qubit pairs. By
	suitable entanglement swapping steps 
	\cite{PhysRevLett.71.4287,PhysRevA.95.032306,TeleportationReview}, the resulting state is a \emph{graph state} \cite{Hein04,Hein06}. Methods for purifying any graph state via measurements and classical communication have been studied \cite{KMBD06} and applications in quantum networks considered \cite{pirker2017modular, Markham}. As already discussed, setting up the shared quantum state before the actual request for communication, is preferable in terms of efficiency of communication, but also allows for detection and prevention of channel or node failure.

	For a given graph state and a request for communication between two distinct nodes, a straightforward solution would be to find a shortest path between the nodes, create a ``repeater'' line (by isolating the path from its environment), and then perform measurements on the intermediate nodes, thereby creating an EPR pair between the two. However, this approach is far from optimal since it requires measuring a large number of nodes and therefore diminishes the secondary use of the residual quantum state. Here, we propose another method that requires at most as many measurements as this ``repeater'' protocol, in general leaving a larger part of the graph state intact, while simultaneously solving bottleneck issues in the network. 
	The proposed method is based on local complementation \cite{Hein04,bouchetLCorbit}  and is already underlying in the prominent bottleneck example of the butterfly scheme \cite{Butterfly,EppingA}. The painful lack of studies
	in this area is due to the fact that local complementation
	does not provide an advantage in classical network coding, since there is no classical equivalent to the application of local Clifford operations in order to achieve serviceable long-range correlations. 
	Finally, we turn towards the problem of extracting graph states from given larger graph states
	via local Clifford operations and Pauli measurements. Using known results from graph theory \cite{VertexMinor, Hoyer, Beigi, RankWidth}, we discuss that while the general problem is known to be {\tt NP}-complete \cite{Dahlberg18}, for specific classes of more structured resources, polynomial time algorithms can be found. All our schemes are based on local complementation,
	but are genuinely quantum, in the way that genuinely multi-partite quantum graph states are manipulated.

	\subsection{Preliminaries}
	A graph $G=(V,E)$ consists of a finite set of vertices $V\subsetneq \mathbb{N}$ and a set $E\subseteq V \times V$ of edges. 
	Vertices that are connected by an edge are called adjacent.
	The set of all vertices that are adjacent to a given vertex $a$ is called the \emph{neighborhood} of $a$ and denoted by $N_a$. 
	We may write $|G|\defeq |V|$ for the number of vertices.
	Graphs have an adjacency matrix with entries
	\vspace{-0.1in}
	\begin{equation}
	\left(\Gamma_G\right)_{i,j}\defeq\begin{cases} 
	1, & \text{if }(i,j)\in E\\
	0, & \text{if }(i,j)\not\in E
	\end{cases}
	\end{equation}
	associated with them. In this work, we only consider \emph{simple} graphs, i.e., graphs that do not contain edges connecting a vertex to itself, or multiple edges between the same pair of vertices. 	Given a graph $G$, we can prepare a graph state vector $\ket{G}$ associated with it as follows. First, a qubit in 
	$\ket{+}=(\ket{0}+\ket{1})/\sqrt{2}$ is prepared for each of the vertices in $V$. 
	Subsequently, a controlled-$Z$ operation is applied to each pair of qubits that is adjacent in $G$. The resulting graph state vector can thus be written as
	\begin{align}\ket{G}\defeq\prod_{(i,j)\in E} CZ_{i,j} \ket{+}^{\otimes V}.
	\end{align}
	It is important to stress that graph states do not have to be prepared in this fashion. In fact, we here anticipate 	the states to be prepared from EPR pairs and entanglement swapping in a quantum network. 
	Note that local Pauli measurements on a graph state result in a different graph state up to local unitary corrections (cf. Proposition 7 in Ref.\ \cite{Hein06}). Here, we will omit these local corrections for the sake of clarity. 
	In this work we will make use of a graph transformation called local complementation. By $\tau_a(G)$ we denote the graph that results from locally complementing $G$ with respect to the vertex $a$. 
	
	\begin{define}[Local complementation]
		A graph $G=(V,E)$ and vertex $a\in V$ define a graph $\tau_a(G)$ with adjacency matrix
		\begin{equation}
		\Gamma_{\tau_a(G)}\defeq\Gamma_G+\Theta_a \mod 2,
		\end{equation}
		where $\Theta_a$ is the complete graph of the neighborhood $N_a$.
	\end{define}
	
	Local complementation on a graph is equivalent to applying local Clifford gates on the respective graph state \cite{VandenNest1}. In particular, the graph state that results from local complementation with respect to node $a$ of the graph state vector $\ket{G}$, is defined by $\ket{\tau_a(G)}\defeq U_a^\tau\ket{G}$, where $U_a^\tau\defeq ({iX_a})^{1/2} ({-i Z_{N_a}})^{1/2}$. 
	It is possible to verify whether two graph states can be transformed into each other via sequential local complementations in polynomial time  \cite{VandenNest2}.
	As we only consider local Clifford operations and Pauli measurements, the resulting states remain graph states and can be described in terms of the pre-measurement graph with the help of local complementations and $Z$-measurements \cite{Hein06}.

\subsection{Reducing the number of measurements}
We have already argued that sharing graph states between the nodes of a network allows for quicker communication with less requirements for channel capacity and memory than sharing EPR pairs between nodes. However, it is not known, given a shared graph state, what the optimal technique for entanglement sharing between nodes that are not connected via physical links is. An approach equivalent to the well-established repeater networks would be to create a ``path'' that connects the two nodes, and then, via entanglement swapping, create a long distance EPR pair. 

In the following, we will prove that a ``repeater'' method is not optimal regarding the number of measurements to be performed. Having a significantly reduced number of measurements is extremely useful in quantum networks, since it allows us to ``extract'' more entanglement from the shared graph state. 
The \emph{repeater protocol} entails first isolating a path between two nodes $a$ and $b$ by $Z$-measuring the neighborhood of said path  (creating a repeater line) and then connecting $a$ to $b$ via $X$-measurements along the intermediate nodes of the path. The \emph{$X$-protocol} is doing the reverse, first $X$-measuring the intermediate nodes on the path between $a$ and $b$, and then $Z$-measuring everything that is left in the neighborhoods of $a$ and $b$ respectively.  We specifically prove the following theorem in the Appendix.

\begin{theorem}[Creating maximally entangled pairs]\label{thm1}
	We can create an EPR pair between two nodes $a$ and $b$ of an arbitrary graph state using the $X$-protocol with at most as many measurements as with the repeater protocol.
\end{theorem}
\vspace{-0.1in}
 \begin{figure}[h]
 	\includegraphics[width=7cm]{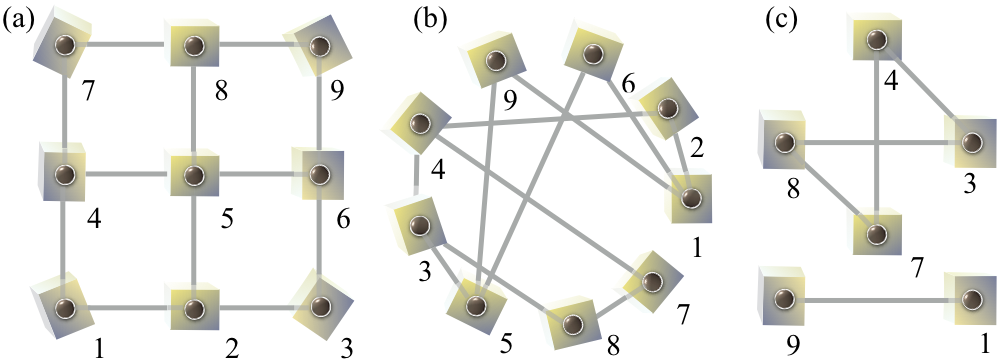}
 	\caption{An EPR pair and a residual graph state are distilled from a cluster state with $9$ qubits	 using the $X$-protocol on the path $(1,2,5,6,9)$. 	This is visualised by considering (b) local complementations with respect to nodes $1,2,5,6,1$, followed by (c) the deletion of nodes $2,5,6$ on the graph that describes the graph state.}
 	\label{fig:9qubitCluster}
 \end{figure}

The proof compares the number of measurements required when running the two different algorithms. 
This technique decreases the number of measurements used in standard repeater scenarios, when we know a pair of nodes that intends to communicate (in this case $a$ and $b$). In particular, it allows for a larger part of the graph state to remain intact for future use. 
Fig.~\ref{fig:9qubitCluster} visualizes how the $X$-protocol for a $9$ qubit cluster state allows us to  communicate between the nodes $1$ and $9$ while keeping a residual  graph state for simultaneous communication between any pair of nodes in $\{3,4,7,8\}$. Here, the residual graph state can be turned into the desired second EPR pair by a single measurement. Note that if we would first isolate the path between nodes $1$ and $9$ and then apply standard repeater protocols, the distillation would require the measurement of at least six nodes and thereby render the extraction of a second EPR pair impossible. 

It is also beneficial to compare our protocol to the standard entanglement swapping methods based on directly sharing EPR pairs over the underlying network. To build the graph states of Fig.~\ref{fig:9qubitCluster}(c) over the underlying grid network using entanglement swapping, we need 12 EPR pairs, which is the same number required to build the cluster state in Fig.~\ref{fig:9qubitCluster}(a). The crucial difference is that, while the cluster state can accommodate more communication requests, the direct generation of the graph states in Fig.~\ref{fig:9qubitCluster}(c) via entanglement swapping limits the communication scenarios that we can implement.

If no more information about future communication requests is available, it is more resource economical 
to choose the shortest path that has the minimal neighborhood. However, the following lemma will be useful in case we would like to allow more than one pair of nodes to communicate simultaneously. Specifically, it gives a visualisation of different possibilities of entanglement generation between nodes.

\begin{lemma}[Equivalence of measurements]\label{lem1}
$X$-measurements along a shortest path between two nodes are equivalent to performing a series of local complementations on the path, followed by $Z$-measurements on the intermediate nodes.
\end{lemma}
\vspace{-0.1in}
\begin{figure}[!h]
\centering
	\includegraphics[width=7cm]{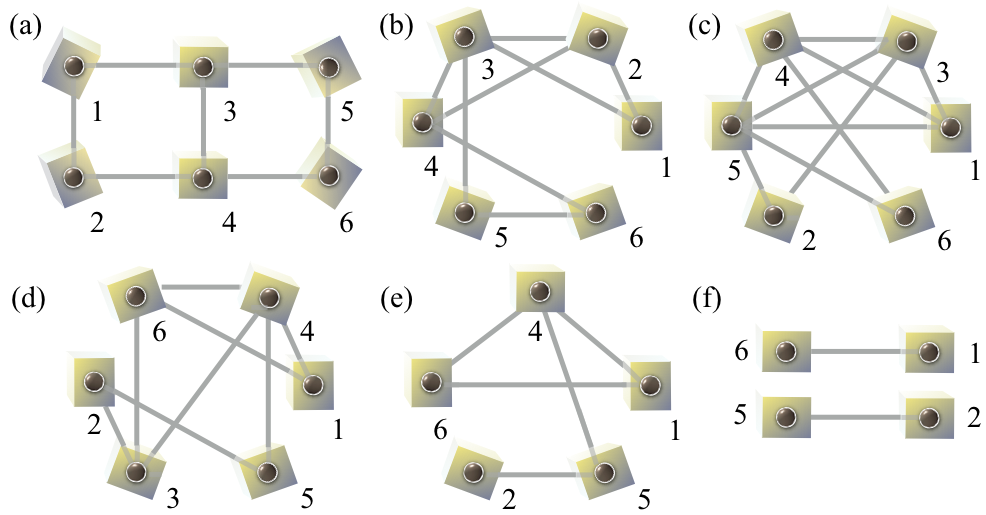}
		\caption{{\bf Establishment of two EPR pairs.} 
		Starting from the butterfly network (a), we perform consecutive local complementations on nodes $1$ (b), $3$ (c) and $4$ (d). 
		The qubits $3$  (e) and $4$ (f) are measured in order to establish EPR pairs between nodes $\{1,6\}$ and $\{2,5\}$.}
		\label{fig:Butterfly}
\end{figure}

Lemma \ref{lem1} allows us to transform the problem of establishing entanglement between nodes into finding suitable graphs by successive local complementations of the network graph. These repeated local complementations generate an orbit, the \emph{LC-orbit} \cite{bouchetLCorbit}. As already mentioned, if the only request for communication is between two nodes, then a shortest path with minimal neighborhood is chosen, in order to minimize the number of measurements. However, if the problem at hand is to connect more than one pair of nodes, the local complementation path will be chosen differently, according to the resulting graph.

Even if not at first apparent, this is the strategy for the well-known butterfly network scheme (Fig.~\ref{fig:Butterfly}), where in order to create EPR pairs between nodes $\{1,6\}$ and $\{2,5\}$, $X$-measurements are done on nodes $3$ and $4$. Via Lemma \ref{lem1} this is equivalent to finding a graph in the LC-orbit of the butterfly, where edges $(1, 6)$ and $(2, 5)$ exist, and no edge between sets $\{1, 6\}$ and $\{2, 5\}$ exists. This graph is found via consecutive local complementations on nodes $1$, $3$ and $4$. A $Z$-measurement on nodes $3$ and $4$ allows to extract the two required EPR pairs (Fig.~\ref{fig:Butterfly}(f)). Note that without the second request for connection of nodes $2$ and $5$, the algorithm might have chosen another path to do $X$-measurements. Similarly, the sequence of subfigures in Fig.~\ref{fig:9qubitCluster} demonstrates the equivalent process for the 9-qubit cluster state. Both this specific communication example and the one presented in the butterfly scheme, create bottlenecks in the network; this is further discussed in the following section.

\subsection{Bottleneck quantum networks}
The butterfly network is of particular interest  when considering bottlenecks in the network. If the nodes can share only one EPR pair over each physical link, one of the butterfly's edges is a bottleneck if we aim to build repeater lines to create entanglement between nodes $\{1,6\}$ and $\{2,5\}$. The above method fulfills the requirement of sharing only one EPR pair per physical link, in order to build the appopriate graph state, and yet solves the communication problem by bypassing the bottleneck in the network. We can further show by exhaustive search that the butterfly network structure is uniquely minimal with respect to the number of nodes.

\begin{proposition}[No bottleneck]
There is no 5-node graph state that has a bottleneck for simultaneous communication between two pairs of nodes and that can be solved using local Cliffords and a Pauli measurement of a single node.
\end{proposition}

\begin{proposition}[Bottleneck]
	There are only four 6-node graph states that have a bottleneck for simultaneous communication between two pairs of nodes and  that can be solved using local Cliffords and Pauli measurements.
	\end{proposition}

The only four possible 6-node graphs mentioned in the above proposition are the ones resulting from node relabeling in Fig.~\ref{fig:Butterfly}(a). Specifically, if we intend to establish EPR pairs between nodes $\{1,6\}$ and $\{2,5\}$, we obtain the four graphs by exchanging labels within the sets $\{3,4\}$ and $\{1,6\}$. Note that in allowing arbitrary local Cliffords and Pauli measurements we considered a wider class of possible algorithms than just the aforementioned $X$-protocol.

\subsection{Obtaining GHZ and other multi-partite resources}

As a further aspect, we now turn to the key question of how to extract resource states such as GHZ states from a given graph state. The more general question, whether from a given graph state vector $\ket{G}$ we can extract another graph state vector $\ket{H}$  via a sequence of local measurements, has recently been proven to be {\tt NP}-complete \cite{Dahlberg18}. This was done by solving a well-known problem in graph theory called the {\tt VERTEX-MINOR} problem, which asks whether from a graph $G$, another graph $H$ can be extracted via a sequence of (i) local complementations and (ii) deletion of vertices. Note that the {\tt NP}-completeness of deciding whether graph $H$ is a vertex-minor of $G$ has been proven for labeled graphs, which are relevant for communication scenarios, since the nodes are distinct. 

Having said that, there are polynomial-time algorithms that solve the problem for important instances. A first relevant instance involves GHZ states \cite{GHZ, wallnofer2016two}, which are essential resources for multi-partite schemes in quantum networks beyond point-to-point architectures, such as \emph{quantum secret sharing}  
\cite{SecretSharing,bell2014experimental}. 
Building upon the method described in Theorem \ref{thm1}, we can show the following corollary.

\begin{corollary}[Extraction of GHZ3 states]\label{cor1}
We can always distill a $3$-partite GHZ state between arbitrary vertices of a connected graph state in polynomial time.
\end{corollary}
In order to obtain a $3$-partite GHZ state, we use a slightly altered version of the $X$-protocol. The proof examines different cases corresponding to distinct relative positions of the three vertices within the  graph and is given in the Appendix. We now propose a sufficient criterion in order to extract $4$-partite GHZ states; note that the extraction of a complete graph of four nodes (which is a graph representing a GHZ4) is thought to be difficult in general \cite{PivotNP}.


\begin{proposition}[Extraction of GHZ4 states]\label{Prop3}
We can distill a $4$-partite GHZ state from a graph state when the underlying graph has a repeater line as vertex-minor, which contains all four nodes of the final GHZ state and at least one extra node between two pairs of the nodes.\label{lem:GHZ4}
\end{proposition}

\begin{figure}[h]
\centering
	\includegraphics[width=7cm]{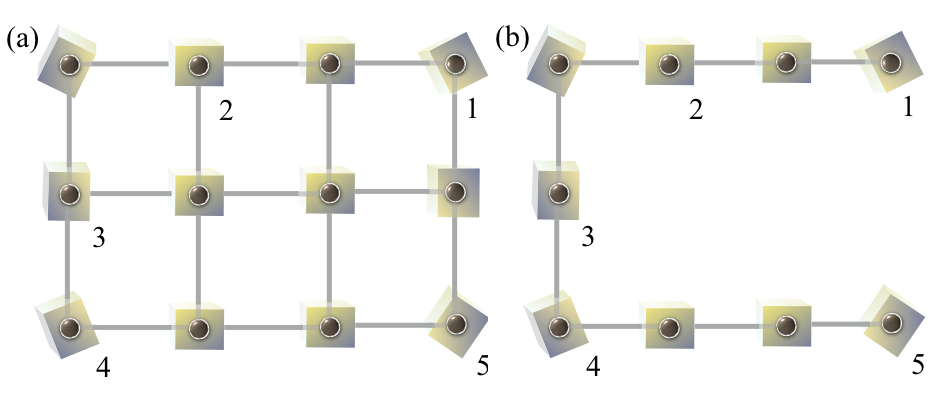}
	\includegraphics[width=7cm]{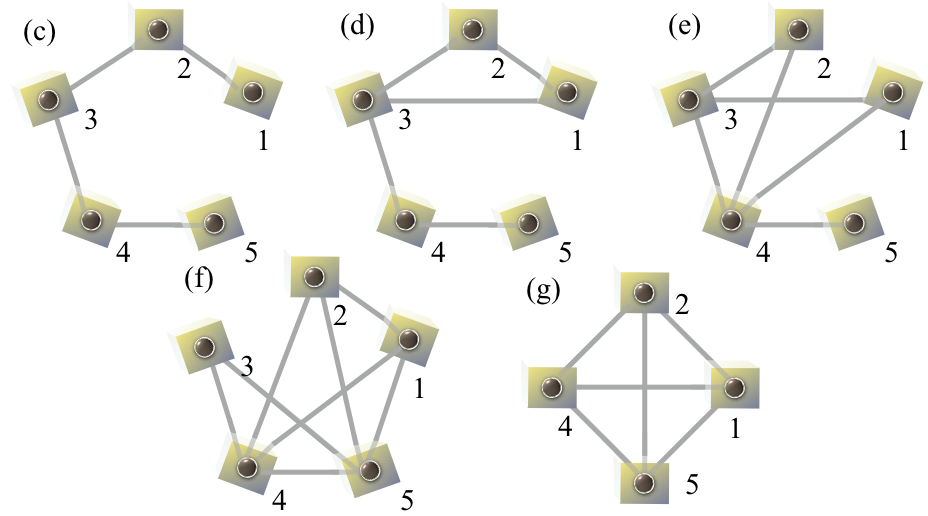}
		\caption{{\bf Prototypical extraction of a GHZ4 state.} 
				In order to distill a GHZ4 state between the nodes $1,2,4$ and $5$ of a $12$ qubit cluster state (a), we $Z$-measure three nodes (b) and $X$-measure the four remaining intermediate nodes, thus isolating a repeater line (c) in accordance with Proposition \ref{Prop3}.
		We then perform consecutive local complementations on nodes $2$ (d), $3$ (e) and $4$ (f). Node $3$  is finally $Z$-measured (g) to establish the desired GHZ4 state.
	}
	\label{fig:GHZ4}
\end{figure}

The required criterion is very likely to be fulfilled for simple network architectures, over which the graph state will be shared. Fig.~\ref{fig:GHZ4} demonstrates this for a short-distance square-grid network, which is used to share a cluster state. Here, Fig.~\ref{fig:GHZ4}(c) visualizes the minimal repeater line that is described in Proposition \ref{Prop3}. The proof of said proposition is given in the Appendix. A more general result using the notion of rank-width \footnote{The \emph{rank-width} $k$ of a graph $G$ is the minimum width of all its rank decompositions. This amounts to $k$ being the smallest integer such that $G$ can be related to a tree-like structure by recursively splitting its vertex set so that each cut induces a matrix of rank at most $k$.  The rank-width is bounded iff the clique-width is bounded \cite{Approximating}. Graphs with rank-width at most one are those where all connected induced subgraphs preserve distance \cite{VertexMinor}.} is based on 
Refs.\
\cite{Dahlberg18,Wehner18,Oum17,courcelleoum,CMR00}.

\begin{observation*}[Extraction of graph states from graph states with bounded rank-width] \label{thm2} For a graph state vector $\ket{G}$ with an underlying graph of bounded rank-width, there exists a poly-time algorithm that decides if a graph state vector $\ket{H}$ can be extracted from $\ket{G}$ using local Clifford operations and $Z$-measurements, and gives the sequence of operations to be applied.
\end{observation*}

For a graph $G$, there exist algorithms with runtime $O(|G|^3)$ \cite{RankWidth} that, for a fixed $k$, either give a rank decomposition of width at most $24k$ or reply that the rank-width is larger than $k$. Then, when such a rank decomposition is given, for a fixed graph $H$, a linear time algorithm can test whether  $H$ is a vertex minor of $G$ and return the sequence of local complementations and vertex deletions to be applied \cite{Wehner18}. Intuitively speaking, many structured graphs have bounded rank-width. E.g., highly sparse random graphs have a bounded rank-width \cite{RandomGraphs}, and so do graphs with a bounded tree-width \cite{oum2008rank}. For those graphs, the above observation readily applies, and it can be decided whether resource states can be extracted. 

\subsection{Discussion}

As an outlook, we mention an exciting link to classical network coding theory:
Schemes have been previously studied for the teleportation of a quantum state from a set of nodes (sources) to another set (sinks), and a connection with classical network coding has been established \cite{Kobayashi09,Kobayashi11,TeleportationReview}. The $k$-pair problem in classical network routing is relevant here, where $k$ sources want to simultaneously send information to $k$ sinks. In subsequent studies, the connection with measurement-based quantum computation \cite{Oneway,gross2007novel}
has been established \cite{debeaudrap14} and subsequently, the question of sharing a general graph state over a network has been addressed \cite{EppingB}. However, the latter work has a shortcoming; the mapping of the network is done using linear codes that require the generation of two-colorable graph states at each node, and it is not straightforward to see how to make this mapping to a given network structure, where each node holds a single qubit. 

In this work, we have discussed the manipulation of multi-partite entangled resources for applications in quantum routing and quantum communication across  quantum networks. A key application of the strategies laid out is
in parallel quantum key distribution and notions of conference key agreement. 
We have seen that via local complementation, quantum routing schemes with a reduced number of measurements
outperforming standard repeater schemes can be found, bottleneck quantum networks can be treated and the
question of extracting multi-partite resources largely addressed.
 It is important to stress that, while these algorithms are classical, they apply
to true multi-partite quantum entangled states. To provide further perspective, also note that since every \emph{stabilizer state} is equivalent to some graph state
\cite{Schlingemann}, the methods laid out here are also expected to be useful in the design of quantum error correcting codes.
It is the hope that this work triggers further studies of
manipulating multi-partite entangled resources for quantum routing, which seem urgently needed in the light of the
rapid experimental progress on quantum networks.

\subsection{Acknowledgements}
We thank S.\ Wehner and A.\ Dahlberg for discussions. A.~P.~acknowledges support from European Union's Horizon 2020 Research and Innovation program under Marie Sklodowska-Curie Grant Agreement No.~705194, J.~E.~support from the BMBF (Q.Com-Q), the 
DFG (EI 519/14-1, EI 519/9-1), the Templeton Foundation, and the ERC (TAQ). Upon completion of this work, we became aware of work similarly motivated \cite{Wehner18,Dahlberg18}.


\section*{Appendix}
\setcounter{theorem}{0}
\begin{theorem}[Creating maximally entangled pairs]\label{thm1_a}
	We can create an EPR pair between two nodes $a$ and $b$ of an arbitrary graph state using the $X$-protocol with at most as many measurements as with the repeater protocol.
\end{theorem}

Before stating the proof of Theorem \ref{thm1} we introduce some additional notation.
For two subsets $A,B \subseteq V$ of vertices  we denote by 
\begin{equation}
E(A,B)\defeq\left\{(a,b):a\in A, b\in B, a\neq b\right\}
\end{equation}
the set of all possible edges between the two sets. Note that $E(A,B)$ in general contains edges that are not contained in the edgeset of an arbitrary given graph $G=(V,E)$.
For a vertex subset $W\subseteq V$ we denote by $E_{|W}$ the subset of $E$ that contains every edge that connects to at least one vertex in $W$.
We may  subtract a set of edges $F$ from $E$. That is, by $E\setminus F$ we denote the set of edges in $E$ that are 
not contained in $F$. 
For such a second set of edges $F$ we also define the \emph{symmetric difference} of $E$ and $F$ as
\begin{equation}
E\Delta F \defeq (E\cup F) \setminus (E\cap F).
\end{equation}
If $F$ happens to be a subset of $E$ the symmetric difference $E\Delta F$ is identical to $E \setminus F$. Otherwise the mutual edges are removed from the union of the two sets.

If the qubit associated with vertex $v\in V$ of a graph state  is $X$-measured, the transformation of the corresponding edge set $E$ can be described in terms of symmetric differences. Independent of a choice $w\in N_v$, the new edge set is given by
\begin{equation}\label{xsymmetricdifference}
\left(E\Delta E_{vw}\Delta E_{v\cap w}\Delta E_{v\setminus w}\right)\setminus E_{|\{v\},} 
\end{equation}
where $E_{vw}\defeq E(N_{w},N_{v})$,  $E_{v\cap w}\defeq E(N_{w}\cap N_{v},N_{w}\cap N_{v})$ and $E_{v\setminus w}\defeq E(\{w\},N_{v}\setminus \{w\})$ are introduced as a shorthand notation. The subtraction of the set containing only the edges that connect to the vertex itself at the end of Eq.~\eqref{xsymmetricdifference} represents the isolation of $v$  due to the measurement. Given a distinct pair of vertices $a,b\in V$,  a \emph{path} of length $k$ from $a$ to $b$ is an ordered list $(v_1,v_2,\ldots, v_k)$ such that $v_1=a$, $v_k=b$, and for all $i\in[k-1]$, vertices $v_i$ and $v_{i+1}$ are adjacent. We denote by $l$ the length of a shortest path from $a$ to $b$ within the graph at hand. 

In the following we will describe  how the neighborhoods $N_{v_i}$ of vertices $v_i$ change due to Pauli measurements on the graph state. To indicate that the graph and therefore some neighborhoods may have changed, we make use of an additional index $t$. By $N^{(t)}_{v_i}$ we denote the neighborhood of node $v_i$ after the $t^{th}$ Pauli measurement on the initially given graph state. We carry this notation over for symmetric differences.  In the expression $E^{(t)}(\cdot, \cdot)$ the $t$ indicates that all involved neigborhoods are regarded after the $t^{th}$ Pauli measurement. From the context it will always be obvious which nodes are measured in which step. In particular
\begin{equation}
N^{(0)}_{v_1}\cup N^{(0)}_{v_2} \cup \ldots \cup N^{(0)}_{v_{k}}
\end{equation}
is the joint neighborhood of a path  $(v_1,v_2,\ldots, v_k)$ in the initially given graph before any measurements are made. 
In our proof we compare two measurement algorithms that both have the goal of establishing an EPR pair between the nodes $a$ and $b$ of a given graph state.
\begin{itemize}
	\item The \emph{repeater  protocol} selects the shortest path connecting $a$ to $b$ that has the minimum combined neighborhood. Every node that lies in the combined neighborhood of this path but not on the path itself is then $Z$-measured. This isolates the path from the rest of the graph creating a repeater line. Finally, every intermediate vertex on the line is $X$-measured yielding the EPR pair between the two nodes.
	\item The \emph{X-protocol} measures the intermediate vertices along the same shortest path  in the $X$ basis . Subsequently, the neighborhoods of the two nodes are $Z$-measured to create the desired EPR pair.  
\end{itemize}

\noindent{}We start our comparison by counting the number of measurements required in the  repeater  protocol. Along the minimal neighborhood path $(v_1,v_2,\ldots, v_l)$ connecting $v_1=a$ to $v_l=b$, every neighboring vertex that does not lie on the path itself is measured in the $Z$ basis. This requires 
$|N^{(0)}_{v_1}\cup N^{(0)}_{v_2} \cup \ldots \cup N^{(0)}_{v_l}|-l$ measurements. To obtain the desired EPR pair from this newly-created repeater line we then measure the vertices $v_2, v_3, \ldots, v_{l-1}$ in the $X$ basis. These $l-2$ measurements remove the intermediate nodes
of the path one by one. The repeater protocol thus requires 
\begin{equation}\label{repeaterprotocolcounting}
|N^{(0)}_{v_1}\cup N^{(0)}_{v_2} \cup \ldots \cup N^{(0)}_{v_{l}}|-2
\end{equation}
Pauli measurements in total to establish the EPR pair.

In order to prove Theorem \ref{thm1}, we will now count the number of measurements required when using the  $X$-protocol. The protocol starts by $X$-measuring along the path $(v_2,\ldots, v_{l-1})$. Here, $t$ indicates the $X$-measurement of node $v_{t+1}$  and $N^{(t)}_{v_i}$ is the neighborhood of $v_i$ after the $t^{th}$ $X$-measurement. We will need the following observation.

\begin{observation*}[Minimizing measurements]\label{equivalentclaim}
	To prove Theorem \ref{thm1}, it suffices to show
	\begin{equation}\label{equivalentEquation}
	N^{(l-2)}_{a}\cup N^{(l-2)}_{b} \subsetneq N^{(0)}_{v_1}\cup N^{(0)}_{v_2} \cup \ldots \cup N^{(0)}_{v_l}
	\end{equation}
	and that we can find at least $l-2$ elements in the neighborhood of the initial (before any measurement) path between $a$ and $b$ that are not contained in the neighborhoods of $a$ and $b$ after the $X$-measurements of the $X$-protocol.
\end{observation*}

\noindent In total, the $X$-protocol requires ($l-2$) $X$-measurements along the shortest path, and subsequently $|N^{(l-2)}_{a}\cup N^{(l-2)}_{b}|-2$
$Z$-measurements on those vertices that have connecting edges to $a$ or $b$ (we need to subtract $a$ and $b$ from the count). From Eq.~\eqref{repeaterprotocolcounting} it follows that Theorem \ref{thm1} holds if 
\begin{equation}
|N^{(l-2)}_{a}\cup N^{(l-2)}_{b}| +l-2 \leq |N^{(0)}_{v_1}\cup N^{(0)}_{v_2} \cup \ldots \cup N^{(0)}_{v_l}|.
\end{equation}

\noindent{}Therefore, in order to prove Theorem \ref{thm1}, it is sufficient to show that Eq.~\eqref{equivalentEquation} is fulfilled. 

\begin{proof}
	
	In the following we will examine how the neighborhoods of $a=v_1$ and 
	$b=v_l$ change with the sequence of $X$-measurements along the shortest path. The first such measurement is at vertex $v_2$. The measurement results in 
	a new graph state with the same set of vertices and with an edge set that can be calculated via a series of symmetric differences, according to Eq.~\eqref{xsymmetricdifference},
	\begin{equation}
	\left(E\Delta E^{(0)}_{v_1v_2}\Delta E^{(0)}_{v_1\cap v_2}\Delta E^{(0)}_{v_2\setminus v_1}\right) \setminus E_{|\{v_2\},} \label{zero}.
	\end{equation}
	By definition of $E(\cdot,\cdot)$ we find
	\begin{align}
		&E^{(0)}_{v_1v_2}=\left\{(x_1,x_2):x_i\in N^{(0)}_{v_i},i =1,2; x_1\neq x_2\right\},\label{one}\\
		&E^{(0)}_{v_1\cap v_2}=\left\{(x,y):x,y\in N^{(0)}_{v_1}\cap N^{(0)}_{v_2} , x\neq y\right\},\label{two}\\
		&E^{(0)}_{v_2\setminus v_1}=\left\{(v_1,x_2):x_2\in N^{(0)}_{v_2}\setminus \{v_1\}\right\}.\label{three} 
	\end{align}
	In the following we analyse the consecutive symmetric differences in 
	Eq.~\eqref{zero} step by step. 
	In particular, we are interested in how the $X$-measurement on $v_2$ changes the neighborhoods of $a=v_1$ and $b=v_l$. 
	Since we have $v_1\in \textcolor{black}{ N^{(0)}_{v_2} }$, the set
	$E^{(0)}_{v_1v_2}$ contains all edges that where connected to $v_1$ before the measurement. 
	From Eq.~\eqref{one} we can thus infer that 
	$N^{(1)}_{v_1}$ does not contain any of the elements that where previously contained in $N^{(0)}_{v_1}$. 
	The second symmetric difference in Eq.~\eqref{zero} does not alter the neighborhood of the starting vertex $v_1$ by virtue of Eq.~\eqref{two} and $v_1$ not being 
	in the intersection of $N^{(0)}_{v_1}$ and $N^{(0)}_{v_2}$. If follows that the only contribution to $N^{(1)}_{v_1}$ comes from Eq.~\eqref{three}. We find
	\begin{equation}
	N^{(1)}_{v_1}=N^{(0)}_{v_2}\setminus \{v_1\}\label{first}
	\end{equation}
	and ascertain that the new graph after the $X$-measurement on $v_2$ has a path $(v_1,v_3,v_4,\ldots,v_l)$ of length $l-1$ connecting
	$a=v_1$ and $v_{l}=b$. 
	We note that this new, shorter path is again a shortest path between $a$ and $b$, since the $X$-measurement only alters the neighborhood of the measured node. The vertex $v_2$ is now isolated, that is, there are no edges that connect it to the remaining graph. The following measurements will remove the
	other intermediate vertices from the path one by one. 
	
	The next $X$-measurement on $v_3$  yields $N^{(2)}_{v_1}=N^{(1)}_{v_{3}}\setminus \{v_1\}$
	and finally after the $t^{th}$ measurement, it holds that
	\begin{equation}
	N^{(t)}_{v_1}=N^{(t-1)}_{v_{t+1}}\setminus \{v_1\}\label{A}.
	\end{equation}
	We now examine how the neighborhood of $v_{t+2}$ is changed by the $t^{th}$ measurement. Before we write down the general expression for $N^{(t)}_{v_{t+2}}$, we consider
	the special case $t=1$, that is, the environment of vertex $v_3$ after the measurement on $v_2$.
	Again, Eq.~\eqref{two} does not contribute to $N^{(1)}_{v_3}$, 
	because $v_3\in N^{(0)}_{v_1}\cap N^{(0)}_{v_2}$ would be a contradiction to $(v_1,v_2,\ldots,v_l)$ being a shortest path before the first measurement. Via Equation
	\eqref{one} we add those elements of $N^{(0)}_{v_1}$ that have not previously been connected to $v_3$ and remove those that where. 
	Compared to $N^{(0)}_{v_{3}}$, the neighborhood $N^{(1)}_{v_{3}}$ also gains the 
	element $v_1$ by virtue of Eq.~\eqref{three}, since $v_3$ is certainly an element of $N^{(0)}_{v_{2}}\setminus \{v_1\}$. 
	To sum up this gives 
	$N^{(1)}_{v_{3}}=\{v_1\}\cup \left(N^{(0)}_{v_{3}}\cup N^{(0)}_{v_{1}}\right) \setminus \left(N^{(0)}_{v_{3}}\cap N^{(0)}_{v_{1}}\right)$ and thus
	
	\begin{equation}\label{B}
	N^{(t)}_{v_{t+2}}=\{v_1\}\cup \left(N^{(t-1)}_{v_{t+2}}\cup N^{(t-1)}_{v_{1}}\right) \setminus \left(N^{(t-1)}_{v_{t+2}}\cap N^{(t-1)}_{v_{1}}\right)\end{equation}
	for the general case after the $t^{th}$ measurement.
	For any $t=2,3,\ldots,l-2$ we can combine Eqs. \eqref{A} and \eqref{B} and obtain the expression
	\begin{equation}
	N^{(t)}_{v_1}=\left(N^{(t-2)}_{v_{t+1}}\cup N^{(t-2)}_{v_{1}}\right) \setminus \left(N^{(t-2)}_{v_{t+1}}\cap N^{(t-2)}_{v_{1}}\right)\label{zwoelf}
	\end{equation}
	In particular we can now write recursive expressions for $N^{(l-2)}_{a}$ and $N^{(l-2)}_{b}$. More specifically, we obtain
	\begin{eqnarray}
	N^{(l-2)}_{v_1}&=&\left(N^{(l-4)}_{v_{l-1}}\cup N^{(l-4)}_{v_{1}}\right) 
	\setminus \left(N^{(l-4)}_{v_{l-1}}\cap N^{(l-4)}_{v_{1}}\right),\nonumber\\
	N^{(l-2)}_{v_l}&=&\{v_1\}\cup \left(N^{(l-3)}_{v_l}\cup N^{(l-3)}_{v_{1}}\right)\setminus \left(N^{(l-3)}_{v_l}\cap N^{(l-3)}_{v_{1}}\right).\nonumber\\\label{start}
	\end{eqnarray}
	Eqs.\ \eqref{start} contain multiple expressions of the type $N^{(t)}_{v_{t+3}}$. These expressions can easily be simplified.
	For instance the case $t=1$ entails $N^{(1)}_{v_{4}}=N^{(0)}_{v_{4}}$, since $v_4\notin N_{v_i}$ for $i=1,2$. Otherwise this would be a contradiction  
	to $(v_1,v_2,\ldots,v_l)$ being the shortest path before the first measurement. By the same argument we recursively get
	\begin{equation}
	N^{(t)}_{v_{t+3}}=N^{(t-1)}_{v_{t+3}}=\ldots=N^{(0)}_{v_{t+3}}
	\end{equation}
	for all $t=1,2,\ldots,l-3$. Together with  Eqs.\ \eqref{zwoelf} and \eqref{start}, this recursively implies that we can write $N^{(l-2)}_{a}$ and $N^{(l-2)}_{b}$ as 
	the union and intersection of sets of the type $N^{(0)}_{v_{i}}$ and $N^{(1)}_{v_{i}}$ where $v_i$ is a vertex on the initial shortest path. For all $v_i$, except for $i=1,2,3$, we
	have $N^{(0)}_{v_{i}} = N^{(1)}_{v_{i}}$, since the negation would imply $(v_1,v_2,\ldots,v_l)$ not being the shortest path before the first measurement. The neighborhood of $v_2$ is empty after the first 
	measurement. Now, Eq.~\eqref{first} and $N^{(1)}_{v_{3}}=\{v_1\}\cup \left(N^{(0)}_{v_{3}}\cup N^{(0)}_{v_{1}}\right) \setminus \left(N^{(0)}_{v_{3}}\cap N^{(0)}_{v_{1}}\right)$
	imply
	\begin{equation}
	N^{(l-2)}_{v_1}\cup N^{(l-2)}_{v_l} \subsetneq N^{(0)}_{v_1}\cup N^{(0)}_{v_2} \cup \ldots \cup N^{(0)}_{v_l}.
	\end{equation}
	The subset relation is proper, because the $l-2$ vertices $v_2,v_3,\ldots,v_{l-1}$ are contained in $N^{(0)}_{v_1}\cup N^{(0)}_{v_2} \cup \ldots \cup N^{(0)}_{v_l}$
	but not in  $N^{(l-2)}_{v_1}\cup N^{(l-2)}_{v_l}$. This implies
	\begin{equation}
	|N^{(l-2)}_{v_1}\cup N^{(l-2)}_{v_l}|+l-4 \leq |N^{(0)}_{v_1}\cup N^{(0)}_{v_2} \cup \ldots \cup N^{(0)}_{v_l}|-2,
	\end{equation}
	which concludes the proof.
\end{proof}

\vspace{0.1in}

\noindent{}Now we will prove Lemma \ref{lem1} and thereby show the equivalence of successive $X$-measurements to $Z$-measurements on a graph in the LC orbit.

\setcounter{lemma}{0}
\begin{lemma}[Equivalence of measurements]\label{lem1}
	$X$-measurements along a shortest path between two nodes are equivalent to performing a series of local complementations on the path, followed by $Z$-measurements on the intermediate nodes.
\end{lemma}

\begin{proof}
	An $X$-measurement of a node is equivalent to locally complementing a neighbor, then locally complementing the actual node and Z-measuring, followed by a final local complementation of the same neighbor. Suppose that the nodes $v_i$, $i=1,\dots, n$ constitute a shortest path. We denote by $X_i$ and $Z_i$ the $X$- and $Z$-measurements on node $i$ respectively, and by $LC_i$ the action of local complementation with respect to the node $v_i$.  Then
	\begin{equation}
	X_2= LC_1~LC_2~Z_2~LC_1
	\end{equation}
	is a valid decomposition of $X_2$ in terms of local complementations and $Z$-measurements.
	If there is no shorter path connecting $v_1$ and $v_3$, this means that the $X_2$ measurement (and more specifically $LC_2$) creates a link between $v_1$ and $v_3$. Therefore, when we measure $X_3$, we can again choose  $v_1$ as a neighbor and find
	\begin{equation}
	X_3= LC_1~LC_3~Z_3~LC_1.
	\end{equation}
	Continuing along the path, we finally find that
	\begin{align}
		X_2\cdots X_{n-1}= LC_1~LC_2~Z_2~LC_3~Z_3~\cdots LC_{n-1}~Z_{n-1}~LC_1,
	\end{align}
	since two consecutive local complementations with respect to the same vertex cancel each other out.
	However, $Z_i$ commutes with $LC_j$ since measuring in $Z$ removes the node and all adjacent edges. If $i\in N_j$, it does not matter whether a local complementation will connect $v_i$ with any other node or not, since all connections will disappear after the measurement. We can therefore push all $Z$-measurements to the end and obtain 
	\begin{align}
		X_2\cdots X_{n-1}= LC_1~\cdots LC_{n-1}~LC_1~Z_2~\cdots ~Z_{n-1}
	\end{align}
	to conclude the proof.
\end{proof}

\vspace{0.1in}
\noindent{}Now, building upon the $X$-protocol, we give a proof of Corollary \ref{cor1} by a short case analysis.

\begin{customcor}{1}[Extraction of GHZ3 states]\label{cor1}
	We can always distill a $3$-partite GHZ state between arbitrary vertices of a connected graph state in polynomial time.
\end{customcor}

\begin{proof}
	Again we take $(a=v_1, v_2, \ldots, v_{l}=b)$ to be a shortest path in the initial graph.
	If $c\in N^{(0)}_{v_1}\cup N^{(0)}_{v_2} \cup \ldots \cup N^{(0)}_{v_{l}}$ lies in the neighborhood of the chosen path, there are two subcases.
	\begin{itemize}
		\item{	If $c=v_i$ for some $i \in \{1,2,\ldots,l\}$,  we measure the vertices $v_2,v_3,\ldots,v_{i-1}$, $v_{i+1},\ldots ,v_{l-1}$ in the $X$-basis. After these $l-3$ measurements every vertex in $N^{(l-3)}_{a}\cup N^{(l-3)}_{b}\cup N^{(l-3)}_{c}\setminus \{a,b,c\}$ is measured in the $Z$-basis.}
		\item{If $c\neq v_i$ for all $i \in \{1,2,\ldots,l\}$, we measure the vertices $v_2,v_3,\ldots ,v_{l-1}$ in the $X$-basis.  After these $l-2$ measurements vertex $c$ is certainly contained in $N^{(l-2)}_{a}\cup N^{(l-2)}_{b}$.
			Every vertex but $c$ from this set is then measured in the $Z$-basis.}
	\end{itemize}
	If $c\notin N^{(0)}_{v_1}\cup N^{(0)}_{v_2} \cup \ldots \cup N^{(0)}_{v_{l}}$, we again measure the nodes $v_2,v_3,\ldots ,v_{l-1}$ in the $X$-basis. Without loss of generality let the shortest path $b=w_1, w_2, \ldots, w_{l'}=c$ be shorter than all the paths  from $a$ to $c$. We continue by measuring all vertices in $N^{(l-2)}_{a}$ in the $Z$-basis followed by $w_2,w_3,\ldots ,w_{l'-1}$ in the $X$-basis. Note that $w_i \in N^{(l-2)}_{a}$ for some $i \in \{2,3,\ldots,l'\}$ would be a contradiction to the shortest path assumption.
	Finally, we measure every vertex but $a$ that lies in the neighborhoods of $b$ and $c$ in the $Z$-basis. All of the above cases result in the desired $3$-partite GHZ state between $a$, $b$ and $c$ independent of the choice of paths.
\end{proof}

\noindent Finally, we turn towards the generation of $4$-partite GHZ states as stated in Proposition \ref{lem:GHZ4}. 	

\setcounter{proposition}{2}
\begin{proposition}[Extraction of GHZ4 states]\label{lem:GHZ4_a}
	We can always distill a $4$-partite GHZ state from graph states when their underlying graph (i) is a repeater line, with at least one extra node between two pairs of the final GHZ4 nodes, or (ii) contains such a line as a vertex-minor.
\end{proposition}

\begin{proof}
	To be consistent with the figure in the main text, suppose that we want to have a GHZ4 state between nodes with labels $1,2,4$ and $5$ in the original graph state. If the underlying graph has a repeater line as a vertex-minor, we may separate it from the remaining graph state via appropriate local complementations and measurements. By local complementation on the path and measurement on the nodes that are not part of the final GHZ4, we can always distill the required state, as seen in the figure in the main text. For a large subset of graph states there is however a more efficient way to generate the desired GHZ4 state in analogy to the $X$-protocol.
	Without isolating the repeater line first, we may perform local complementations with respect to  nodes $2,3,4$ followed by $Z$-measurements on $3$ and on every vertex that is connected to any of $1,2,4,5$ and does not lie on the repeater-line itself. If, for example, there is no shorter path connecting $1$ to $5$ initially than the one given by the repeater line, the successive local complementations result in graphs that have a subgraph like the ones displayed in the figure in the main text.\\
\end{proof}

\COMMENT{
	\subsection{Workspace}
	
	\je{Here some work space. First note that there is a figure of merit that, as a one-sided test, allows to 
		check whether two graphs states are LC equivalent. This the cut-rank. If the graphs of 
		two graph states have a different cut-rank, then they are LC inequivalent \cite{CutRank,Hein04}. The converse is not 
		true, two isomorphic Peterson graphs provide a counterexample \cite{Peterson}. Comment added:
		Nice, but no longer relevant.}
	
	\subsubsection{Examples for the line}
	No 4-partite GHZ state on the ending nodes of a repeater line?
	No overlapping EPR pairs on a repeater line? \je{Shall we settle this question before submission?
		Not sure.}
	
	\subsubsection{Examples for the ring}
	Evenly distributed 4-partite GHZ state on the ring with $4k, k=2,3,\ldots$ nodes?
	Can we show that the ``repeater'' scheme is the best we can do, for the 2-pair setting?
	\je{Cute. But good for follow-up work.}
	
	\subsubsection{Ideas for 4-partite GHZ states}
	The following corollary is in some sense related to Theorem \ref{thm2}; having no shorter path between nodes $1$ and $5$ than the path 
	$1,2,3,4,5$ is essentially a local sparsity property.
	
}

\end{document}